\documentclass[twocolumn, twoside]{article}

\usepackage{moreverb}

\usepackage[bookmarksopen,bookmarksnumbered]{hyperref}

\usepackage{amsmath}
\usepackage{amsthm}
\usepackage{amsfonts}
\usepackage{epsfig}
\usepackage{mathcomSTEP,stmaryrd}

\theoremstyle{remark}
\newtheorem{rem}[thm]{Remark}
\theoremstyle{definition}

\begin{document}

\title{
On the Capacity of the Cognitive Interference Channel with a Common Cognitive Message
}

\author{
Stefano Rini$^\diamond$ and Carolin Huppert$^\dagger$ \\[.5cm]
$^\diamond$ Institute for Communications Engineering, Technische Universit\"at M\"unchen, Germany, \\ email: {\tt stefano.rini@tum.de}\\[.2cm]
$^\dagger$ Institute of Communications Engineering,  Ulm University, Germany, \\ email: {\tt carolin.huppert@uni-ulm.de}
}

\date{}

\twocolumn[
  \begin{@twocolumnfalse}
    \maketitle

\begin{abstract}
In this paper the cognitive interference channel with a common message, a variation of the classical cognitive interference channel in which the cognitive message is decoded at both receivers, is studied.
For this channel model new outer and inner bounds are developed as well as new capacity results for both the discrete memoryless and the Gaussian case.
The outer bounds are derived using bounding techniques originally developed by Sato for the classical interference channel and Nair and El Gamal for the broadcast channel.
A general inner bound is obtained combining rate-splitting, superposition coding and binning.
Inner and outer bounds are shown to coincide in the ``very strong interference'' and the ``primary decodes cognitive'' regimes.
The first regime consists of channels in which there is no loss of optimality in having both receivers decode both messages while in the latter
regime interference pre-cancellation at the cognitive receiver achieves capacity.
Capacity for the Gaussian channel is shown to within a constant additive gap and a constant multiplicative factor.
\end{abstract}
\vspace*{.5cm}

  \end{@twocolumnfalse}
]

\section{Introduction}
Cognitive networks are transmission networks where the message of one user is known at multiple nodes.
The study of cognitive networks was inspired by newfound abilities of smart radios  to overhear the transmission taking place over the channel and gather information about neighboring nodes \cite{goldsmith_survey}.
The information theoretical study of cognitive networks has so far focused on small networks with a limited number of users and messages.
A classical such model is the cognitive interference channel \cite{devroye_IEEE}: a channel where two sets of transmitter/receiver pairs communicate over a shared channel, thus interfering with each others' transmission. One of the encoders in the network--the primary transmitter-- has knowledge of only the message to be transmitted to its intended receiver while the other encoder--the cognitive transmitter-- has knowledge of both messages.
The additional knowledge available at the cognitive transmitter models a smart and adaptable device which is able to acquire the primary message
from previous or simultaneous transmissions in the network.

The cognitive channel has been studied in depth in the last few years and many results have been derived for this model.
The largest known inner bound is provided in \cite{rini2009state} and is obtained using classical random coding techniques such as rate-splitting, superposition coding and binning.
The most general outer bound is derived in \cite{maric2005capacity} by using an argument originally devised for the broadcast channel in \cite{NairGamal06}.
Capacity for both the memoryless and the Gaussian case is not known in general but only for specific subclasses.
For the memoryless channel, the largest known region where capacity is known is the ``better cognitive decoding'' regime, \cite{RTDjournal1}, where capacity is achieved by rate-splitting the cognitive message in a private and a public part and decoding the primary message at both receivers.
This result generalizes two previous results: a ``very strong interference'' result, \cite{MaricUnidirectionalCooperation06}, and a ``very weak interference'' result, \cite{WuDegradedMessageSet}.
In the ``very strong interference'' regime there is no loss of optimality in having both receivers decode both messages. This results is akin to the ``very strong interference'' result for the interference channel \cite{sato.IFC.strong,CostaElGamal87}.
For the ``very weak interference'' regime, instead, capacity is achieved by having the cognitive receiver decode the interference while the primary receiver treats the interference as noise.

Capacity is also known for the semi-deterministic cognitive interference channel, \cite{Rini:ICC2010}, that is for the channel in which the output at the cognitive receiver is a deterministic function of inputs while the output at the primary decoder is a any random function.
Here capacity is achieved by having both cognitive and the primary message private and pre-coding for the cognitive transmission against the primary interference.

A larger set of capacity results is available for the Gaussian case than for the discrete memoryless case.
Capacity is known in the ``weak  interference'' regime \cite{WuDegradedMessageSet}, a regime that contains the ``very weak interference'' regime.
As for the ``very weak interference'' regime, the optimal strategy for the primary receiver is to treat the interference as noise but, in this case, the cognitive codeword is pre-coded against the interference created by the primary transmission.
Capacity for the Gaussian case is also known in the ``primary decodes cognitive'' regime of \cite{RTDjournal2},
in which the cognitive message is decoded at both receivers and pre-coded against the interference created by the primary user at the cognitive decoder.
It must be noted that a crucial tool to achieve capacity for both the deterministic and the Gaussian channel is interference pre-cancellation using binning as in the classical Gel'fand Pinsker problem \cite{GelFandPinskerClassic}.
For these two classes of channels binning at the cognitive transmitter can fully remove the effect of the interference experienced at the cognitive receiver. This property does not hold for a general channel and makes it easier to prove capacity.

Capacity for the Gaussian case is also known to within a constant additive gap of 1~bit/s/Hz and to within a multiplicative factor of two \cite{Rini:ICC2010}.
That is, the gap between the inner and the outer bound can be bounded by a constant difference as well as by a constant ratio.
The first result well characterizes the capacity region in the high SNR regime while the latter gives a good capacity approximation for the low SNR regime.

Despite of the difficulty in deriving capacity for the cognitive interference channel, capacity is fully known for a simple variation of the cognitive interference channel: the cognitive interference channel with a degraded message set \cite{liang2009capacity}.
In this channel the cognitive receiver is required to decode both the cognitive and the primary message.
The capacity achieving strategy is have a public primary message and split the cognitive message in a public and a private parts.  The cognitive public codeword is then superposed over the primary public one and the private cognitive codewords over the two other public codewords.
Since the primary message is decoded at both receivers, no interference pre-cancellation is required at the cognitive transmitter.

\subsection*{Contributions}

In this paper we study the cognitive interference channel with a common cognitive message, a variation of the cognitive interference channel where the primary receiver decodes both messages.
We derive inner and outer bounds for this channel model as well as new capacity results.
Some of the techniques used to derive these results are similar to the techniques used in
\cite{maric2005capacity,RTDjournal1,RTDjournal2} for the cognitive interference channel.
In particular we highlight the relationship between this channel model and the cognitive interference channel in the ``strong interference'' regime,
a regime where there is no loss of generality in having the primary receiver to also decode the cognitive message.
Capacity in this regime is known only in a subset of the parameter regime and progress in improving either inner or outer bounds has been slow.

We first derive a series of outer bounds, each containing an increasing number of auxiliary random variables.
The simpler outer bounds are easy to evaluate but are not tight in general while  outer bounds with more auxiliary random variables are tighter but harder to evaluate and compare to inner bounds.

We introduce an inner bound that employs rate-splitting,  superposition coding and binning which is very much reminiscent of the achievable scheme in \cite{rini2009state}.
A compact representation of this achievable region is obtained by considering \emph{rate-sharing}. 
Rate sharing consists in  ``shifting'' rate contributions from the cognitive rates to the primary rates as well as from public rates to private rates.
The Fourier Motzkin elimination of the different rate contributions provides  a more compact expression of the original achievable region.
This techniques was introduced by Hajek and Pursley \cite{hajek1979evaluation}  when deriving an achievable region for the broadcast channel with common messages and successively employed when in \cite{liang2007rate} to derive an achievable region for the relay broadcast channel.

Inner and outer bounds are shown to coincide for the general channel model in the ``very strong interference'' and the ``primary decodes cognitive'' regimes.
The first regime is related to a result available for the cognitive interference channel and we present it for completeness. The ``primary decodes cognitive'' regime, however, have not been investigated for the general cognitive interference channel so far.
We also prove capacity for the semi-deterministic channel, that is the channel where the channel output at the cognitive decoder is a deterministic function of the channel inputs while the output at the primary receiver is any random function.
As for the cognitive interference channel, this result relies on the fact that for deterministic channels binning can completely pre-cancel the interference experienced at the receiver.

For the Gaussian case we prove capacity to within a constant additive gap of 1 bit/s/Hz and a multiplicative factor 2.
The proof for the constant additive gap offers an alternative proof to the result in \cite{Rini:ICC2010} although it does not improve on the overall gap.

Most of the results we derive are shown using an interesting transmission scheme in which the cognitive message, decoded at both receivers, is also
pre-coded against the interference experienced at the cognitive decoder.
The pre-coding of the cognitive message does not allow the primary decoder to reconstruct the interfering signal.
The cognitive message acts instead as a side information at the primary receiver when decoding its intended message.
The paper concludes with a set of numerical simulations that compare outer bounds, to show the rate advantages of different transmission choices and 
illustrate the regimes where capacity is known for the Gaussian case.

\subsection*{Paper Organization}
The paper is organized as follows:
Sec.~\ref{sec:Channel Model} introduces the considered channel model. 
In  Sec \ref{sec:An Outer Bounds for the Cognitive Interference Channel with a Common Cognitive Message} we introduce outer bounds to the capacity region.
while inner bounds are presented in Sec.  \ref{sec:Inner Bounds for the Cognitive Interference Channel with a Common Cognitive Message}. 
Sec.~\ref{sec:Capacity for the Cognitive Interference Channel with a Common Cognitive Message} contains the capacity results for the general channel model.
Sec.~\ref{sec:The Gaussian Cognitive Interference Channel with a Common Cognitive Message} presents results for the Gaussian case.
In Sec.~\ref{sec:Numerical Simulations} we  introduce a series of numerical simulation of the results presented in the paper.
Sec.~\ref{sec:conclusion} concludes the paper.

\section{Channel Model}
\label{sec:Channel Model}

\begin{figure}
\centering
\includegraphics[width=1.0\columnwidth]{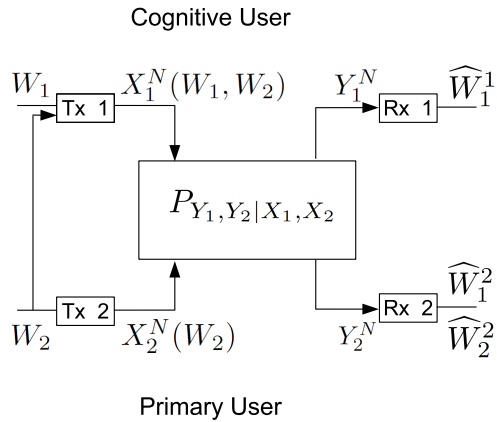}
\caption{The Cognitive InterFerence Channel with Common Cognitive Message (CIFC-CCM).}
\label{fig:channelModel}
\end{figure}
The Cognitive InterFerence Channel with a Common Cognitive Message (CIFC-CCM), as shown in Fig.~\ref{fig:channelModel}, is obtained from the classical Cognitive InterFerence Channel (CIFC)  by having the primary receiver decode both messages.
It consists of two transmitter-receiver pairs that exchange independent messages over a common channel. In the CIFC-CCM transmitter $i$, $i\in\{1,2\}$, has input alphabet $\Xcal_i$ and its receiver has output alphabet $\Ycal_i$. The channel is assumed to be memoryless with transition probability $P_{Y_1,Y_2|X_1,X_2}$ so that $P_{Y_1,Y_2|X_1,X_2}^N$ indicates the memoryless extension to the channel transition probability over $N$ channel uses.
Encoder~2 wishes to communicate a message $W_2$ uniformly distributed on
$\Wcal_2 = [1: 2^{N R_2}]$ to decoder~2 in $N$ channel uses at rate $R_2$.
Similarly, encoder~1, wishes to communicate a message $W_1$ uniformly distributed on
$\Wcal_1 = [1: 2^{N R_1}]$ to both decoder~1 and decoder~2 in $N$ channel uses at rate $R_1$.
Encoder~1 (i.e., the cognitive user) knows its own message $W_1$ and the one of encoder~2 (i.e., the primary user), $W_2$.
A rate pair $(R_1,R_2)$ is achievable if there exist sequences of encoding functions
\begin{align*}
X_1^N &= f_{X_1^N}(W_1, W_2), \;\;   f_{X_1^N} : \Wcal_1 \times \Wcal_2 \rightarrow {\cal X}_1^N, \\
X_2^N &= f_{X_2^N}(W_2), \;\; \;\;\;\;\;\;  f_{X_1^N} : \Wcal_2 \rightarrow {\cal X}_2^N,
\end{align*}
with corresponding sequences of decoding  functions
\begin{align*}
\widehat{W}_1^1 & =  f_{\widehat{W}_1^1}(Y_1^N), \;\;   f_{\widehat{W}_1^1} : {\cal Y}_1^N \rightarrow \Wcal_1, \\
\widehat{W}_1^2 & =  f_{\widehat{W}_1^2}(Y_2^N),   \;\; f_{\widehat{W}_1^2} : {\cal Y}_2^N \rightarrow \Wcal_1, \\
\widehat{W}_2^2 & =  f_{\widehat{W}_2^2}(Y_2^N),   \;\; f_{\widehat{W}_2^2} : {\cal Y}_2^N \rightarrow \Wcal_2,
\end{align*}
such that
\begin{equation*}
\underset{N \rightarrow \infty}{\lim} {\rm Pr} \left\{\left[\widehat{W}_1^1, \widehat{W}_1^2, \widehat{W}_2^2\right] \neq \left[W_1, W_1, W_2 \right]\right\}=0.
\end{equation*}

The capacity region is defined as the closure of the union of achievable $(R_1,R_2)$ pairs
 \cite{ThomasCoverBook}.
Standard strong-typicality is assumed; properties may be found in
 \cite{kramerBook}.

\smallskip

In the following we focus in particular on the Gaussian CIFC-CCM as depicted in Fig.~\ref{fig:GaussianchannelModel}.
For this class of channels, the input/output relationship
is:
\eas{
Y_1& =X_1 + \ a X_2 + Z_1,  \\
Y_2& =X_2 + |b| X_1 + Z_2,
}{\label{eq:in/out gaussian CIFC-CCM}}
for $a,b \in \Cbb$ and for $Z_i \sim \Ncal_{\Cbb}(0,1)$, where the $\Ncal_{\Cbb}$ indicates complex, circularly symmetric jointly Gaussian Random Variables (RVs).
Moreover, the channel inputs are subject to the second moment constraints
\ea{
\Ebb \lsb \labs X_i  \rabs^2 \rsb \leq P_i, \ \ \ i \in \{1,2\}.
\label{eq:power constraint}
}
A channel where the outputs are obtained from a linear combination of the inputs plus an additional complex Gaussian term can be reduced to
the formulation in \eqref{eq:in/out gaussian CIFC-CCM} and \eqref{eq:power constraint} without loss of generality \cite[App. A]{RTDjournal2}.
Note that the coefficient $b$ can be taken to be real and positive without loss of generality, for this reason it will be indicated as $|b|$ in the remainder of the paper \cite[App. A]{RTDjournal2}.

\begin{figure}
\centering
\includegraphics[width=1.0\columnwidth]{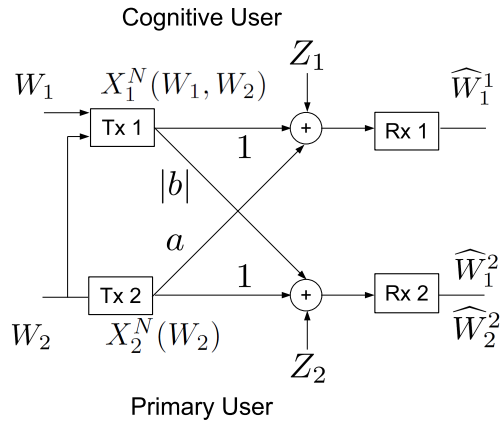}
\caption{The Gaussian Cognitive Interference Channel with Common Cognitive Message (Gaussian CIFC-CM).}
\label{fig:GaussianchannelModel}
\end{figure}

\section{Outer Bounds}
\label{sec:An Outer Bounds for the Cognitive Interference Channel with a Common Cognitive Message}

In this section we develop three outer bounds to the capacity region: the first outer bound does not contain any auxiliary random RV, the second bound one and the third bound three auxiliary RVs.
Additionally, each outer bound contains the previous bounds as a special cases  but the increasing amount of auxiliary RVs makes it harder to evaluate in closed form and compare to the inner bounds.

\smallskip

The first outer bound for the capacity region of the general CIFC-CCM is based on results known for the cognitive interference channel in the ``strong interference'' regime.
This outer bound contains no auxiliary RV and thus can be easily evaluated in closed form.

\begin{thm}{\bf An Outer Bound for the CIFC-CCM}
\label{th:Outer Bound for the CIFC-CCM}

Any achievable region for the CIFC-CCM is contained in the region
\eas{
R_1         & \leq I(Y_1 ; X_1 | X_2 ),
\label{eq:Outer Bound for the CIFC-CCM R1} \\
R_1         & \leq I(Y_2 ; X_1 | X_2 ),
\label{eq:Outer Bound for the CIFC-CCM R1 weak} \\
R_1 + R_2   & \leq I(Y_2 ; X_1 , X_2),
\label{eq:Outer Bound for the CIFC-CCM sum rate 1}
}{\label{eq:Outer Bound for the CIFC-CCM}}
union over all the joint distributions
\ea{
P_{X_1 X_2} P_{Y_1 Y_2 | X_1,X_2}.
}
\end{thm}

\begin{proof}
The outer bound in \eqref{eq:Outer Bound for the CIFC-CCM R1} was originally devised for the classical CIFC in \cite{WuDegradedMessageSet} and is valid for the
CIFC-CCM as well, since the cognitive decoder is decoding only the cognitive message.
The bound \eqref{eq:Outer Bound for the CIFC-CCM R1 weak}  is obtained from Fano's inequality  as
\pp{
& N R_1 - N \ep_{N}  \\
& \quad \quad  \leq I(Y_2^N ; W_1) \\
& \quad \quad  = I(Y_2^N ; W_1 | W_2) \\
& \quad \quad  = \sum_{ k=1}^N  H(Y_{2,k} | Y_{2, k+1}^N , W_2 , X_{2,k}) \\
& \quad \quad  \quad \quad  - H(Y_{2,k} | Y_{2, k+1}^N , W_1 , W_2, X_{1,k} , X_{2,k} ) \\
& \quad \quad  \leq \sum_{ k=1}^N  H(Y_{2,k} |  X_{2,k}) -  H(Y_{2,k} | X_{1,k} , X_{2,k} ) \\
& \quad \quad  \leq N I(Y_{2,Q}; X_{1,Q} |  X_{2,Q}, Q),
}
where $Q$ is the time sharing RV, uniformly distributed in the interval $[1:N]$.

The bound \eqref{eq:Outer Bound for the CIFC-CCM sum rate 1} is derived in a similar fashion:
\pp{
& N (R_1 + R_2 ) - N \ep_N  \\
& \quad \quad \leq I(Y_2^N ; W_1 , W_2)  \\
& \quad \quad  = \sum_{k=1}^N H(Y_{2, k} | Y_{2 , k+1}^N)\\
& \quad \quad  \quad \quad  - H(Y_{2, k} | Y_{2 , k+1}^N, W_1, W_2 , X_{1,k} , X_{2,k}) \\
& \quad \quad  \leq \sum_{k=1}^N I(Y_{2,k}; X_{1,k}, X_{2,k})\\
& \quad \quad  = N I(Y_{2,Q}; X_{1,Q}, X_{2,Q} |  Q).
}
All the bounds are decreasing in the time sharing RV $Q$ and thus it can be dropped.
\end{proof}

\begin{rem}
\label{rem:storng int outer bound}
The bound \eqref{eq:Outer Bound for the CIFC-CCM R1 weak} is redundant if
\ea{
I(X_1 ; Y_1 | X_2) \leq I(X_1 ; Y_2 |X_2),
\label{eq:strong interference}
}
for all the distributions $P_{X_1,X_2}$.
Condition \eqref{eq:strong interference} corresponds to the condition describing the ``strong interference'' regime for the CIFC.
Thus, when dropping \eqref{eq:Outer Bound for the CIFC-CCM R1 weak} from the outer bound in Th.~\ref{th:Outer Bound for the CIFC-CCM}, one obtains the ``strong interference'' outer bound for the CIFC \cite{maric2005capacity}.
This outer bound  is capacity in the ``very strong interference'' regime for the general CIFC and is capacity in the ``primary decodes cognitive'' regime for the Gaussian CIFC.

\end{rem}

Rem.~\ref{rem:storng int outer bound} formally defines the relationship between the CIFC in ``strong interference'' and the CIFC-CCM.
For a CIFC in the ``strong interference'' regime  there is no loss of optimality in having the primary receiver decodes both messages.
Under condition \eqref{eq:strong interference}, the rate of the cognitive message is not bounded by the decoding capabilities of the primary receiver.
For these reasons, the capacity of the CIFC is equivalent to the one of the CIFC-CCM when condition \eqref{eq:strong interference}  holds.
We avoid referring to condition \eqref{eq:strong interference}  as ``strong interference'' condition for the CIFC-CCM as one cannot properly define ``interference'' in this case since the primary receiver is decoding both the cognitive message and the primary message.

\smallskip

Next, we derive an outer bound inspired by the outer bound in \cite{WuDegradedMessageSet} which is known to be tight for the CIFC in the ``very weak interference'' regime.

\begin{thm}{\bf An Outer Bound for the CIFC-CCM with One RV}
\label{th:An Outer Bound for the CIFC-CCM with one RV}

Any achievable region for the CIFC-CCM is contained in the region
\eas{
R_1         & \leq I(Y_1 ; X_1 | X_2 ),\label{eq:R_1 oneRV 1}\\
R_1         & \leq I(Y_2 ; X_1 | X_2 ), \label{eq:R_1 oneRV 2}\\
R_2         & \leq I(Y_2; U, X_2), \label{eq:R_2 oneRV}\\
R_1 + R_2   & \leq I(Y_2 ; U, X_2) + I(Y_1; X_1 | U , X_2) \label{eq:R_1+R_2 oneRV 1}\\
R_1 + R_2   & \leq I(Y_2 ; X_1 , X_2),\label{eq:R_1+R_2 oneRV 2}
}{\label{eq:Outer Bound for the CIFC-CCM oneRV}}
union over all the joint distributions $P_{U,X_1,X_2}.$
\end{thm}
\begin{proof}
The bounds \eqref{eq:R_1 oneRV 1}, \eqref{eq:R_1 oneRV 2} and \eqref{eq:R_1+R_2 oneRV 2} are from Th.~\ref{th:Outer Bound for the CIFC-CCM}. The two bounds \eqref{eq:R_2 oneRV} and \eqref{eq:R_1+R_2 oneRV 1} are derived in \cite{WuDegradedMessageSet} for the general CIFC and are also valid in the CIFC-CCM.
\end{proof}

\begin{rem}
\label{rem:reduce outer bound 2 to outer bound one}
The outer bound in Th.~\ref{th:Outer Bound for the CIFC-CCM} can be obtained from the outer bound in Th.~\ref{th:An Outer Bound for the CIFC-CCM with one RV} by dropping \eqref{eq:R_2 oneRV} and \eqref{eq:R_1+R_2 oneRV 1}. The region obtained by dropping these two bounds necessarily contains the region in \eqref{eq:Outer Bound for the CIFC-CCM oneRV} but the inclusion is not granted to be strict.

\end{rem}

\begin{thm}{\bf An Outer Bound for the CIFC-CCM with Three RVs}
\label{th:An Outer Bound for the CIFC-CCM with three RV}

Any achievable region for the CIFC-CCM is contained in the region
\eas{
R_1         & \leq I(Y_1 ; X_1 | X_2 ),\label{eq:R_1 twoRV 1}\\
R_1         & \leq I(Y_2 ; X_1 | X_2 ), \label{eq:R_1 twoRV 2}\\
R_1         & \leq I(Y_1 ; V, U_1 ), \label{eq:R_1 twoRV 3}\\
R_1         & \leq I(Y_2 ; V, U_1 ), \label{eq:R_1 twoRV 4}\\
R_2         & \leq I(Y_2 ; V, U_2 ), \label{eq:R_2 twoRV}\\
R_1 + R_2   & \leq I(Y_2 ; X_2 | U_1, V) + I(Y_1; V, U_1), \label{eq:R_1+R_2 twoRV 1}\\
R_1 + R_2   & \leq I(Y_1 ; X_1 | U_2, V) + I(Y_2; V, U_2), \label{eq:R_1+R_2 twoRV 2}\\
R_1 + R_2   & \leq I(Y_2 ; X_1 , X_2),\label{eq:R_1+R_2 twoRV 3}
}{\label{eq:An Outer Bound for the CIFC-CCM with two RV}}
union over all the joint distributions $P_{V,U_1,U_2,X_1,X_2}$ that factor as
\ea{
P_{U_1} P_{U_2} P_{V |U_1 ,U_2}P_{X_2 | U_2 V} P_{X_1 | U_1 , U_2}.
\label{eq:factorization maric type}
}
\end{thm}

\begin{proof}
The bounds \eqref{eq:R_1 twoRV 1}, \eqref{eq:R_1 twoRV 2} and \eqref{eq:R_1+R_2 twoRV 3} are from Th.~\ref{th:Outer Bound for the CIFC-CCM}. Bound \eqref{eq:R_1 twoRV 3}, \eqref{eq:R_2 twoRV},
\eqref{eq:R_1+R_2 twoRV 1} and \eqref{eq:R_1+R_2 twoRV 2} are derived in \cite{MaricGoldsmithKramerShamai07Eu} for the general CIFC with a restriction to input distributions where $X_2$ is a function of $(U_2, V)$ and $X_1$ is a function of $(U_1, U_2, V)$.
Bound \eqref{eq:R_1 twoRV 4} is induced by the common cognitive message and derived according to \eqref{eq:R_1 twoRV 3} in \cite{MaricGoldsmithKramerShamai07Eu}.
\end{proof}

\begin{rem}
The factorization in \eqref{eq:factorization maric type} differs from the factorization in \cite[Th. 4]{MaricGoldsmithKramerShamai07Eu} and is, in general, more restrictive. The tightening of the factorization can be obtained by noting that the RVs $U_1$ and $U_2$ are associated with the messages $W_1$ and $W_2$ respectively.
\end{rem}

\begin{rem}
The outer bound in Th.~\ref{th:An Outer Bound for the CIFC-CCM with one RV} can be obtained from the outer bound in Th.~\ref{th:An Outer Bound for the CIFC-CCM with three RV} by considering equations \eqref{eq:R_1 twoRV 1}, \eqref{eq:R_1 twoRV 2}, \eqref{eq:R_2 twoRV}, \eqref{eq:R_1+R_2 twoRV 2} and \eqref{eq:R_1+R_2 twoRV 3} and letting $U=[U_2, V]$.
The region obtained by dropping these five bounds necessarily contains the region in \eqref{eq:An Outer Bound for the CIFC-CCM with two RV} but the inclusion is not granted to be strict.
\end{rem}

\begin{cor}{\bf BC-DMS Outer Bound}
\label{cor:BC-DMS outer bound}
Any achievable region for the CIFC-CCM is contained in the region

\eas{
R_1         & \leq I(Y_1 ; U),\\
R_1 + R_2   & \leq I(Y_2 ; X_1, X_2 | U) + I(Y_1 ; U), \\
R_1 + R_2   & \leq I(Y_2 ; X_1 , X_2),
}{\label{eq:BC-DMS outer bound}}

union over all joint distributions $P_{U,X_1,X_2}$.
\end{cor}

\begin{proof}
The outer bound in Cor.~\ref{cor:BC-DMS outer bound}  can be obtained from the outer bound in Th.~\ref{th:An Outer Bound for the CIFC-CCM with three RV} by considering only Equations \eqref{eq:R_1 twoRV 3}, \eqref{eq:R_1+R_2 twoRV 1}, and \eqref{eq:R_1+R_2 twoRV 3} and letting $U = [V,U_1]$ and $X = U_2 = [X_1,X_2]$.
\end{proof}

\begin{rem}
The outer bound in Cor.~\ref{cor:BC-DMS outer bound} equals the capacity region of the general broadcast channel with degraded message set (BC-DMS) from \cite{KornerMarton_generalBC}.
If full transmitter cooperation is assumed in the CIFC-CCM, i.e., if the cognitive message is also known at transmitter 2, its capacity region equals the BC-DMS capacity.
Note that this region gives an outer bound on the capacity region of the general CIFC in the strong interference regime, cf. \cite{RTDjournal2}.
\end{rem}

\section{Inner Bounds}
\label{sec:Inner Bounds for the Cognitive Interference Channel with a Common Cognitive Message}

In this section we develop an inner bound for the CIFC-CCM that is obtained by a combination of random coding techniques such as rate-splitting, superposition coding and binning.
The primary message is rate-split in three parts: a common part, a private part and a private part broadcasted by the cognitive transmitter to the primary receiver.
The private primary codeword is superposed to the private public one and the cognitive message is superposed over the common primary codeword and also binned against the primary private codeword.

The achievable region can be obtained using standard random coding techniques and is very much reminiscent of the achievable region in \cite{rini2009state}.
We successively show that this region is also achievable by considering \emph{rate-sharing}, that is by allowing part of the cognitive message to be embedded into the primary message and allowing part of the public messages being integrated into the private ones.
Finally, we provide a series of simple achievable regions containing at most one auxiliary RV.
These regions have a simpler expression than the general inner bound and can thus be more easily compared to the available outer bounds. Furthermore, they can be easily used for numerical simulations.

\begin{thm}{\bf A General Inner Bound for the CIFC-CCM}
\label{th:Inner Bounds for the CIFC-CCM}

The following region is achievable for a general CIFC-CCM
\eas{
R_1 & \leq I(Y_1 ; U_{1c} | U_{2c}) - I(U_{1c}; X_2 | U_{2c} ), \\
R_1 & \leq I(Y_2 ; U_{1c}, U_{2pb} | U_{2c}, X_2), \\
R_1 + R_2 & \leq I(Y_2 ; U_{2c}, X_2, U_{1c}, U_{2pb}), \\
R_1 + R_2 & \leq I(Y_1 ; U_{1c}, U_{2c})\notag\\
	  &\;\;\;+ I(Y_2 ; X_2 , U_{2pb} | U_{1c}, U_{2c}),\\
2 R_1 + R_2 & \leq I(Y_1 ; U_{1c} , U_{2c})\notag\\
	  &\;\;\;+ I(Y_2 ; U_{1c}, X_2 , U_{2pb} | U_{2c})\notag\\
	  &\;\;\;- I(U_{1c}; X_2 | U_{2c} ),\label{eq:inner bound 2R1+R2}
}{\label{eq:inner bound}}
for any distribution that factors as
\ea{
P_{U_{1c},U_{2c},U_{2pb},X_1,X_2}.
}
\end{thm}

\begin{proof}
The achievable region in \eqref{eq:inner bound} is obtained in a similar manner than the region in \cite{rini2009state} and using an identical notation.
The message $W_1$ is embedded in the codeword $U_{1c}^N$, where $c$ denotes ``common''.
As mentioned above, the message $W_2$ is rate-split in three parts $W_{2c}, W_{2p}$ and $W_{2pb}$.
$W_{2c}$ is the common part to be decoded by both receivers, $W_{2p}$ is the private part of $W_2$ encoded by both transmitters while $W_{2pb}$ is the private part of the primary message transmitted by the cognitive transmitter.
The rate-splits $W_{2c}$ and $W_{2pb}$ are embedded in the codewords $U_{2c}^N$ and $U_{2pb}^N$ while $W_{2p}$ is mapped directly in the channel output $X_2^N$.
We then consider the achievable region obtained by superposing the codewords $U_{1c}^N$ and $X_2^N$ over $U_{2c}^N$, by superposing $U_{2pb}^N$ over all the other codewords and by binning $U_{1c}^N$ against $X_2^N$.
With this encoding scheme we achieve the region:
\eas{
\Ro_{1c} & \geq I(U_{1c}; X_2| U_{2c})  \nonumber \\
R_{2c} + R_{1c} + \Ro_{1c}  & \leq I(Y_1 ; U_{2c}, U_{1c}),
\label{eq:achievable region before FME a}\\
R_{1c} + \Ro_{1c}  & \leq I(Y_1 ; U_{1c} | U_{2c}),
\label{eq:achievable region before FME b}\\
R_{2c}+R_{1c}+\Ro_{1c}+R_{2p}&+R_{2pb} \leq\notag\\
 I(Y_2 ; U_{2c}, U_{1c}, X_2 , U_{2pb} ) &+ I(U_{1c}; X_2 | U_{2c}),
\label{eq:achievable region before FME c}\\
 R_{1c}+\Ro_{1c}+R_{2p}+R_{2pb} & \leq\notag\\
 I(Y_2 ; U_{1c}, X_2 , U_{2pb} | U_{2c})  &+ I(U_{1c}; X_2 | U_{2c}),
 \label{eq:achievable region before FME d}\\
R_{2p}+R_{2pb} & \leq\notag\\
 I(Y_2 ;  X_2 , U_{2pb} | U_{2c}, U_{1c}) &+ I(U_{1c}; X_2 | U_{2c}),
\label{eq:achievable region before FME e}\\
 R_{1c}+\Ro_{1c}+R_{2pb} & \leq\notag\\
 I(Y_2 ; U_{1c} , U_{2pb} | U_{2c}, X_2) &+ I(U_{1c}; X_2 | U_{2c}),
 \label{eq:achievable region before FME f}\\
 R_{2pb} & \leq I(Y_2 ;  U_{2pb} | U_{2c}, U_{1c}, X_2).
 \label{eq:achievable region before FME g}
}{\label{eq:achievable region before FME}}
The region in \eqref{eq:inner bound} is obtained from the Fourier-Motzkin elimination (FME), \cite{lall-advanced}, of the region in \eqref{eq:achievable region before FME} using the rate-splitting equation
\ea{
R_2 = R_{2c}+R_{2p}+R_{2pb}.
\label{eq:rate splitting definition}
}
The region in  \eqref{eq:inner bound} differs from the region in \cite{rini2009state} in that there is no private cognitive message and in that there are additional constraints on the correct decoding of $U_{1c}^N$ at the primary receiver.
\end{proof}

The chain graph representation, \cite{rini2011achievable}, of the achievable scheme in Th.~\ref{th:Inner Bounds for the CIFC-CCM} is provided in Fig \ref{fig:CGRAS}:
the boxes represent codewords associated with the primary message while the diamond represents the cognitive message. Solid lines represent superposition coding and dotted lines binning.

\begin{figure}
\centering
\includegraphics[width=1.0\columnwidth]{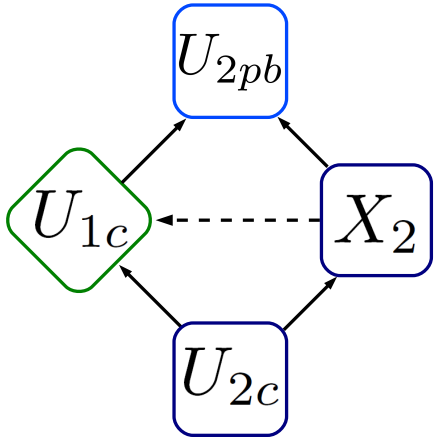}
\caption{The chain graph representation of the achievable scheme in Th.~\ref{th:Inner Bounds for the CIFC-CCM}.}
\label{fig:CGRAS}
\end{figure}

In the following we show that the region in Th.~\ref{th:Inner Bounds for the CIFC-CCM} is also achievable by considering that part of the cognitive message can be embedded into the primary messages and part of the primary public message into the primary private ones. We refer to this technique  as ``rate-sharing'', in lack of any specific term used in literature.

Rate-sharing was introduced by Hajek and Pursley \cite{hajek1979evaluation}  when deriving an achievable region for the broadcast channel with common messages
and successively employed when in \cite{liang2007rate} to derive an achievable region for the relay broadcast channel.
It consists in embedding part of a message after rate-splitting in another codeword: this can be done if rate is shifted from a common message to a private message and from the cognitive message to a primary message.

\begin{cor}{\bf Rate-Sharing}
\label{cor:rate-sharing}
The region in Th.~\ref{th:Inner Bounds for the CIFC-CCM} is also achieved by means of rate sharing.
\end{cor}
The proof of Cor.~\ref{cor:rate-sharing} is given in the Appendix \ref{app:Proof of Cor rate-sharing}.

\begin{rem}
The region in \eqref{eq:inner bound} can also be derived with the approach in \cite{RTDjournal2} which is similar to the simplification of the Han and Kobayashi region for the interference channel in \cite{chong2008han}.
In this approach, one needs to show that certain bounds obtained through the FME are redundant when considering the union over all the possible input distributions.
The proof using rate-sharing is perhaps more laborious but is more powerful in that it does not require one to identify the redundant bounds.
\end{rem}

We now provide two simpler sub-schemes that can be obtained by removing rate-splitting in the scheme of Th.~\ref{th:Inner Bounds for the CIFC-CCM}.

\begin{cor}{\bf Achievable Region Applying only Superposition Coding}
\label{cor:scheme D}
\label{cor:Only superposition coding}

The inner bound in Th.~\ref{th:Inner Bounds for the CIFC-CCM} for $R_{2p}=R_{2pb}=0$ becomes
\eas{
R_1 & \leq I(Y_1 ; X_1 | X_2),
\label{eq:scheme D R1 bound} \\
R_1 & \leq I(Y_2 ; X_1 | X_2),
\label{eq:scheme D R1 Y2 bound} \\
R_1+R_2  & \leq I(Y_1 ; X_1 , X_2),
\label{eq:scheme D Y1 bound} \\
R_1+R_2  & \leq I(Y_2 ; X_1 , X_2),
\label{eq:scheme D Y2 bound}
}{\label{eq:scheme D}}
for any distribution $P_{X_1,X_2}$.
\end{cor}

The achievable scheme in Cor.~\ref{cor:Only superposition coding} employs only superposition coding and both messages are decoded at both receivers.

\begin{cor} {\bf Achievable Region Applying only Binning}
\label{cor:Scheme E}

The inner bound in Th.~\ref{th:Inner Bounds for the CIFC-CCM} for $R_{2c}=R_{2pb}=0$ becomes
\eas{
R_1 & \leq I(Y_1 ; U_{1c}) - I(U_{1c}; X_2),
\label{eq:scheme E R1 bound }\\
R_1 & \leq I(Y_2 ; X_1| X_2),
\label{eq:scheme E R1 bound 2}\\
R_1 + R_2 & \leq I(Y_2 ; X_1, X_2),
\label{eq:scheme E sum rate Y1}\\
R_1 + R_2 & \leq I(Y_1; U_{1c}) + I(Y_2 ; X_2 | U_{1c}),
\label{eq:scheme E extra sum rate Y1}
}{\label{eq:scheme E}}
for any distribution $P_{U_{1c},X_1,X_2}$.
\end{cor}
The achievable scheme in Cor.~\ref{cor:Scheme E} employs only binning and the primary message is private.

\section{Capacity Results}
\label{sec:Capacity for the Cognitive Interference Channel with a Common Cognitive Message}

In this section we derive two capacity results for subsets of the general CIFC-CCM and capacity for the semi-deterministic CIFC-CCM.

\subsection{The Very Strong Interference Regime}
We begin by proving capacity in the ``very strong interference'' regime, a regime where there is no loss of optimality in having both receivers decode both messages.  This regime is reminiscent of the ``very strong interference'' regime for the IFC \cite{sato.IFC.strong} and the CIFC \cite{MaricUnidirectionalCooperation06}. There, the capacity of the channel reduces to the capacity of the compound Multiple Access Channel (MAC) obtained by considering the intersection of the capacity region of the two MAC channels where each decoder decodes both messages.

\begin{thm}{\bf Capacity in the ``Very Strong Interference'' Regime}
\label{th:Capacity in the Very Strong Interference Regime}

If
\ea{
I(Y_2 ; X_1 , X_2) \leq  I(Y_1 ; X_1 , X_2),
\label{eq:very strong interference}
}
the region in \eqref{eq:Outer Bound for the CIFC-CCM} is capacity.
\end{thm}
\begin{proof}
Consider the achievable region in Cor.~\ref{cor:Only superposition coding}:
under condition \eqref{eq:very strong interference} the rate bound \eqref{eq:scheme D Y1 bound} can be dropped and the resulting achievable region coincides with the outer bound in \eqref{eq:Outer Bound for the CIFC-CCM}.
\end{proof}

\begin{rem}
The ``very strong interference'' regime for the CIFC is defined by condition \eqref{eq:very strong interference} and \eqref{eq:strong interference}.
However, condition \eqref{eq:strong interference} is not required to prove capacity for the CIFC-CCM.
\end{rem}
Capacity in the ``very strong interference'' regime  for the CIFC is achieved by having both receivers decode both messages and by superposing the cognitive message over the primary message \cite{maric2005capacity}.
This strategy achieves capacity for a class of CIFC-CCM that we also term  ``very strong interference'' regime.
This definition is not fully accurate since the primary receiver decodes both messages, but is coherent with the CIFC literature.

\subsection{The Primary Decodes Cognitive Regime}
\label{sec:The primary decodes cognitive regime}

The following result shows capacity for a class of channels in which binning allows full interference cancellation at the cognitive decoder.
This result is inspired by the ``primary decodes cognitive'' capacity result available for the Gaussian CIFC \cite{Rini:Allerton2010}.

\begin{thm}{\bf Capacity in the ``Primary Decodes Cognitive'' Regime}
\label{th:Capacity in the Primary Decodes Cognitive Regime}

If
\eas{
I(Y_1 ; U) & \geq I(Y_2 ; U),
\label{eq:PDC DCM}\\
I(U ; X_2 |Y_1) & = 0
\label{eq:PDC binning}
}
for all the distributions  $P_{U,X_1,X_2}$ that factor as
\ea{
P_U P_{X_2} P_{X_1 |U, X_2},
}
the capacity for a CIFC-CCM is given by the region in \eqref{eq:Outer Bound for the CIFC-CCM}
union over all the distributions $P_{X_1,X_2}$.

\end{thm}
\begin{proof}
Consider the scheme in \eqref{eq:scheme E}. For a binning scheme with perfect interference cancellation where \eqref{eq:PDC binning} holds, the inner bound \eqref{eq:scheme E R1 bound } achieves the first outer bound \eqref{eq:Outer Bound for the CIFC-CCM R1}
\ea{
&I(Y_1 ; U_{1c}) - I(U_{1c}; X_2)\notag \\
&= I(Y_1 ; U_{1c} , X_2) - I(Y_1 ; X_2 | U_{1c}) - I(U_{1c}; X_2)\notag \\
%&= I(Y_1 ; X_2) + I(Y_1 ; U_{1c} | X_2)\notag \\
%&\quad \quad   - I(Y_1 ; X_2 | U_{1c}) - I(U_{1c}; X_2)\notag \\
&= I(Y_1 ; U_{1c} | X_2) + I(Y_1 ; X_2) - I(X_2 ; Y_1 , U_{1c})\notag\\
&\quad \quad  + I(X_2 ; U_{1c}) - I(U_{1c}; X_2)\notag \\
&= I(Y_1 ; U_{1c} | X_2) - I(X_2 ; U_{1c} | Y_1)\notag \\
&= I(Y_1 ; X_1 | X_2).
}
The bound in \eqref{eq:scheme E extra sum rate Y1} can be rewritten as:
\ea{
&I(Y_1; U_{1c}) + I(Y_2 ; X_2 | U_{1c})\notag \\
&= I(Y_1; U_{1c}) + I(Y_2 ; X_2 , U_{1c}) - I(Y_2; U_{1c})\notag \\
&= I(Y_1; U_{1c}) + I(Y_2 ; X_1 , X_2) -  I(Y_2; U_{1c})\notag \\
&= I(Y_1; U_{1c}) + I(Y_2 ; X_2 | X_1) -  I(X_2; U_{1c}),
\label{eq:scheme E extra sum rate Y1 new}
}
while \eqref{eq:scheme E sum rate Y1} gives
\ea{
&I(Y_2 ; X_1, X_2) = I(Y_2 ; X_1) + I(Y_2 ; X_2 | X_1)\notag \\
&= I(Y_2; U_{1c}) -  I(X_2; U_{1c}) + I(Y_2 ; X_2 | X_1).
\label{eq:scheme E sum rate Y1 new}
}
This scheme achieves capacity if condition \eqref{eq:PDC DCM} holds since \eqref{eq:scheme E extra sum rate Y1} is redundant in this case as it can be seen from Equations~\eqref{eq:scheme E extra sum rate Y1 new} and \eqref{eq:scheme E sum rate Y1 new}.
\end{proof}

\subsection{Capacity for the Semi-Deterministic Channel}
\label{sec:Capacity for the Semi-Deterministic Cognitive Interference Channel with a Common Cognitive Message}

The semi-deterministic CIFC-CCM is a general CIFC-CCM where the channel output at the cognitive decoder is a deterministic function of the channel inputs, i.e.,
\ea{
Y_1 = f_{Y_1} (X_1,X_2)
\label{eq:det condition Y_1}
}
while the primary output is any random function of the inputs.
When condition \eqref{eq:det condition Y_1} holds, binning at the cognitive transmitter can fully pre-cancel the effect of the interference at the cognitive receiver thus making \eqref{eq:Outer Bound for the CIFC-CCM R1} achievable.

\begin{thm}{\bf Capacity of the Semi-Deterministic CIFC-CCM}
\label{th:Capacity of the semi deterministic CIFC-CCM}

The capacity of the semi-deterministic channel is
\eas{
R_1 & \leq  H(Y_1|X_2), \label{eq:R1 semi deterministic CIFC-CCM 1}\\
R_1 & \leq  I(Y_2 ; X_1 |  X_2), \label{eq:R1  semi deterministic CIFC-CCM 2}\\
R_1 + R_2 & \leq I(Y_2; X_1, X_2), \label{eq:Rsum semi deterministic CIFC-CCM}
}{\label{eq:Capacity of the semi deterministic CIFC-CCM}}
union over all the distributions $P_{X_1,X_2}$.
\end{thm}

\begin{proof}
Consider the transmission scheme given in Cor.~\ref{cor:Scheme E}
to obtain the region \eqref{eq:scheme E}.
For the assignment $U_{1c}=Y_1$, which is possible given \eqref{eq:det condition Y_1},  the bound \eqref{eq:scheme E extra sum rate Y1} is larger than \eqref{eq:scheme E sum rate Y1}
\ea{
& I(Y_1; U_{1c}) + I(Y_2 ; X_2 | U_{1c})\notag\\
& = H(Y_1)+ H(Y_2 | Y_1)  +H(Y_2 | X_1, X_2)  \nonumber \\
& = I(Y_2 ; X_1, X_2) + H(Y_1 |Y_2)  \geq I(Y_2 ; X_1, X_2)
}
and the inner bound in \eqref{eq:scheme E} coincides with
\eqref{eq:Capacity of the semi deterministic CIFC-CCM}
which is also equivalent to the outer bound.
\end{proof}

Note that the result in Th.~\ref{th:Capacity of the semi deterministic CIFC-CCM} does not require the $f_{Y_1}$ to be invertible as for the classical result for the deterministic IFC by El Gamal and Costa \cite{elgamal_det_IC}.

\begin{cor}{\bf Capacity of the Semi-Deterministic CIFC}

The region in \eqref{eq:Capacity of the semi deterministic CIFC-CCM} determines capacity also for a semi-deterministic CIFC in the
``strong interference'' regime.

\end{cor}

\begin{proof}
An outer bound for the CIFC in the ``strong interference'' regime is given by dropping \eqref{eq:Outer Bound for the CIFC-CCM R1 weak} from the outer bound in Th.~\ref{th:Outer Bound for the CIFC-CCM}. The remaining two bounds are given by \eqref{eq:R1 semi deterministic CIFC-CCM 1} and \eqref{eq:Rsum semi deterministic CIFC-CCM} in the semi-deterministic CIFC. In the ``strong interference'' regime of the CIFC defined by Condition~\eqref{eq:strong interference}
the bound given by \eqref{eq:R1 semi deterministic CIFC-CCM 2} is redundant. Thus, the scheme presented in \eqref{eq:scheme E} also achieves capacity in the ``strong interference'' regime of the semi-deterministic CIFC.
\end{proof}

\section{The Gaussian Channel}
\label{sec:The Gaussian Cognitive Interference Channel with a Common Cognitive Message}

We now specialize the results of the previous sections to the Gaussian channel in \eqref{eq:in/out gaussian CIFC-CCM} and derive new capacity results for this channel model. In particular we prove capacity to within a constant gap of one bit and a factor two, thus providing a tight bound for the capacity region at both high and low SNR.

\subsection{Outer Bounds for the Gaussian Case}
\label{sec:Outer bounds for the Gaussian Case}

\begin{cor}{\bf An Outer Bound for the Gaussian CIFC-CCM}
\label{cor:An Outer Bound for the Gaussian CIFC-CCM}

Any achievable region for the Gaussian CIFC-CCM is contained in the region
\eas{
R_1         & \leq  \Ccal(\al \min\{1 , |b|^2\}  P_1),
\label{eq:Outer Bound for the G-CIFC-CCM R1} \\
R_1 + R_2   & \leq \Ccal( P_2 + b^2 P_1 + 2 \sqrt{ \alb |b|^2 P_1P_2}),
\label{eq:Outer Bound for the G-CIFC-CCM sum rate 1}
}{\label{eq:Outer Bound for the G-CIFC-CCM}}
for $\Ccal(x)=\log(1+x)$ and $\alb = 1-\al \in [0,1]$.
\end{cor}
\begin{proof}
The outer bound is obtained from Th.~\ref{th:Outer Bound for the CIFC-CCM} by noting that complex, circularly symmetric channel inputs maximize all the rate bounds simultaneously.
\end{proof}

\begin{cor}{\bf BC-DMS-Based Outer Bound for the Gaussian CIFC-CCM}
\label{cor:BC-DMS-based outer Bound for the Gaussian CIFC-CCM}

Any achievable region for the Gaussian CIFC-CCM is contained in the region
\eas{
R_1&         \leq \notag\\
&\Ccal\left(\frac{\al_1 P_1 + |a|^2 \al_2 P_2 + 2 \Re\{a^* \rho_1\} \sqrt{\al_{12} P_{12}}}{1 + \alb_1 P_1 + |a|^2 \alb_2 P_2 + 2 \Re\{a^*  \rho_2\} \sqrt{\alb_{12} P_{12}}}\right),
\label{eq:BC-DMS outer Bound for the G-CIFC-CCM R1} \\
R_1& + R_2   \leq \notag\\
&\Ccal\left(\frac{\al_1 P_1 + |a|^2 \al_2 P_2 + 2 \Re\{a^* \rho_1\} \sqrt{\al_{12} P_{12}}}{1 + \alb_1 P_1 + |a|^2 \alb_2 P_2 + 2 \Re\{a^* \rho_2\} \sqrt{\alb_{12} P_{12}}}\right)\notag\\
	    & \;\;\;\;+ \Ccal\left( \alb_1 |b|^2 P_1 + \alb_2 P_2 + 2 \Re\{\rho_2\} \sqrt{\alb_{12} |b|^2 P_{12}}\right),
\label{eq:BC-DMS outer Bound for the G-CIFC-CCM sum rate 1} \\
R_1& + R_2   \leq \notag\\
  &\Ccal\left( |b|^2 P_1 + P_2 + 2 \Re\{\rho_1 \sqrt{\al_{12}} + \rho_2 \sqrt{\alb_{12}}\} \sqrt{|b|^2 P_{12}}\right)
\label{eq:BC-DMS outer Bound for the G-CIFC-CCM sum rate 2}
}{\label{eq:BC-DMS outer Bound for the G-CIFC-CCM}}

over the union over all $(\al_1, \al_2, \rho_1, \rho_2)$ satisfying
\ea{
(\al_1, \al_2, |\rho_1|, |\rho_2|) \in [0,1]^4: \;\; \left|\rho_1 \sqrt{\al_{12}} + \rho_2 \sqrt{\alb_{12}}\right| \leq 1,
\label{eq:BC-DMS outer Bound for the G-CIFC-CCM parameters}
}
where $P_{12} = P_1 P_2$,$\al_{12} = \al_1 \al_2$ and $\alb_{12} = \alb_1 \alb_2$.

\end{cor}
\begin{proof}
The Gaussian BC-DMS-based outer bound is obtained from Cor.~\ref{cor:BC-DMS outer bound} as shown in \cite[App. D.B]{RTDjournal2}.
\end{proof}

\subsection{Inner Bounds for the Gaussian Case}
\label{sec:Inner bound for the Gaussian Case}

In the following, the schemes given in Thm.~\ref{th:Inner Bounds for the CIFC-CCM}, Cor.~\ref{cor:scheme D} and Cor.~\ref{cor:Scheme E} are specialized to the Gaussian case.

\begin{cor}{\bf Region Achievable Applying only Superposition Coding in the Gaussian CIFC-CCM}
\label{cor:Region achievable applying only superposition coding in the Gaussian CIFC-CCM}

The region
\eas{
R_1	& \leq \Ccal  \lb \al P_1\rb,
\label{eq:schemeD gaussian R1} \\
R_1 & \leq \Ccal  \lb \al |b|^2 P_1\rb,
\label{eq:schemeD gaussian extra R1} \\
R_1 + R_2 & \leq \Ccal\lb P_1 + |a|^2 P_2  + 2 \Re\{a\} \sqrt{\alb P_1 P_2}\rb,
\label{eq:schemeD gaussian R2} \\
R_1+R_2 & \leq \Ccal \lb |b|^2 P_1 + P_2  +2 \sqrt{\alb |b|^2 P_1 P_2 }\rb,
\label{eq:schemeD gaussian R1+R2}
}{\label{eq:schemeD gaussian}}
is achievable with the assignment of the RVs for the region in Th.~\ref{th:Inner Bounds for the CIFC-CCM}
\eas{
U_{1c} & \sim \Ncal(0, \al P_1),\\
U_{2c} & \sim \Ncal(0, P_2), \\
X_2    & = U_{2c},\\
X_1 & = U_{1c}+ \sqrt{\alb} \sqrt{\f{P_1}{P_2}} X_2,\\
U_{2p}, U_{2pb} & = \emptyset.
}

\end{cor}

\begin{cor}{\bf Region Achievable Applying only Binning in the Gaussian CIFC-CCM}

The region
\eas{
R_1	& \leq f \lb a + \sqrt{\f{\alb P_1}{P_2}},1 ; \la \rb,
\label{eq:schemeE gaussian R1} \\
R_1 & \leq \Ccal  \lb \al |b|^2 P_1\rb,
\label{eq:schemeE gaussian extra R1} \\
R_1 + R_2 & \leq \Ccal\lb |b|^2 P_1 + P_2  + 2 \sqrt{\alb |b|^2 P_1 P_2}\rb  \nonumber \\
 & \quad \quad  + f \lb a + \sqrt{\f{\alb P_1}{P_2}},1 ; \la \rb  \nonumber  \\
  & \quad \quad - f \lb \f 1 {|b|} + \sqrt{\f{\alb P_1}{P_2}}, \f 1 {|b|^2} ; \la \rb,
\label{eq:schemeE gaussian R2} \\
R_1+R_2 & \leq \Ccal \lb |b|^2 P_1 + P_2  +2 \sqrt{\alb |b|^2 P_1 P_2 }\rb,
\label{eq:schemeE gaussian R1+R2}
}{\label{eq:schemeE gaussian}}
for
\eas{
f (&h, \sgs; \lambda)=\notag\\
&\log \lb \f{\sgs+ \al P_1 }{ \sgs + \f{\al P_1 |h|^2 P_2}{\al P_1 + |h|^2P_2+ \sgs} \labs \f{\la}{\la_{\rm Costa}(h,\sgs)}-1\rabs^2 } \rb,
}
and
\ea{
\la_{\rm Costa} (h, \sgs)= \f{\al P_1}{\al P_1 +\sgs}h,
\label{eq:la costa}
}
is achievable with the assignment
\eas{
X_i & \sim \Ncal(0, P_i) \quad i \in \{1,2\}, \\
U_{2c},U_{2pb}& = \emptyset,\\
U_{1c}& = X_1 + \lambda a X_2.
}

\end{cor}

The achievable region in \eqref{eq:schemeE gaussian} reduces for $\lambda = 0$ and $|b| > 1$ to
\eas{
R_1 \leq &\Ccal\left(\frac{\alpha P_1}{\left(\sqrt{\alb P_1}+a\sqrt{P_2}\right)^2+1}\right),\\
R_1 + R_2 \leq &\Ccal\left(\frac{\alpha P_1}{\left(\sqrt{\alb P_1}+a\sqrt{P_2}\right)^2+1}\right)\notag\\
 &+ \Ccal\left(\left(|b| \sqrt{\alb P_1}+\sqrt{P_2}\right)^2\right),
\label{eq:schemeE lam0 gaussian R2}\\
R_1+R_2 & \leq \Ccal \lb |b|^2 P_1 + P_2  +2 \sqrt{\alb |b|^2 P_1 P_2 }\rb.
\label{eq:schemeE lam0 gaussian R1+R2}
}{\label{eq:schemeE lam0 gaussian}}

Note that for the region where \eqref{eq:schemeE lam0 gaussian R1+R2} is a looser bound than \eqref{eq:schemeE lam0 gaussian R2}, the achievable scheme is similar to the capacity achieving scheme in the Gaussian CIFC ``weak interference'' regime of \cite{WuDegradedMessageSet} in which cognitive and primary users are switched.

\subsection{Capacity in the ``Very Strong Interference'' Regime and the ``Primary Decodes Cognitive'' Regime}
\label{sec:Capacity in the Very Strong Interference Regime}

The following corollary states the result of Th.~\ref{th:Capacity in the Very Strong Interference Regime} for the Gaussian case in \eqref{eq:in/out gaussian CIFC-CCM}.

\begin{cor}{\bf Capacity for the Gaussian CIFC-CCM in the ``Very Strong Interference'' Regime}
\label{cor:Capacity for the Gaussian CIFC-CCM in the very strong interference regime}

If
\ea{
(|a|^2-1)P_2 -(|b|^2-1)P_1 - 2 |a - |b||\sqrt{P_1 P_2} \geq 0,
\label{eq:very strong interference gaussian}
}
the capacity of the Gaussian CIFC-CCM is given  by \eqref{eq:Outer Bound for the G-CIFC-CCM}.
\end{cor}
\begin{proof}
Condition \eqref{eq:very strong interference gaussian} is derived from \eqref{eq:very strong interference} for the Gaussian model in \eqref{eq:in/out gaussian CIFC-CCM}. Details can be found in \cite[App. B]{RTDjournal2}.
\end{proof}

The following corollary extends the ``primary decodes cognitive'' regime of \cite{RTDjournal2} to the Gaussian CIFC-CCM:

\begin{cor}{\bf The ``Primary Decodes Cognitive Interference'' Regime for the Gaussian CIFC-CCM}
\label{cor:Primary Decodes Cognitive Interference Regime for the CIFC-CCM}

If
\eas{
P_2 \labs 1- a |b|\rabs^2 \geq &(|b|^2-1)(1 + P_1 + |a|^2 P_2 ) \notag\\
&- P_1 P_2 \labs 1 - a |b|\rabs^2, \\
P_2 \labs 1- a |b|\rabs^2 \geq &(|b|^2-1)(1 + P_1 + |a|^2 P_2 \notag\\
&+2 \Re\{a\}\sqrt{P_1 P_2 }),
}{\label{eq:PDC}}
then \eqref{eq:Outer Bound for the G-CIFC-CCM} is the  capacity of the Gaussian CIFC-CCM.
\end{cor}
\begin{proof}
Consider the scheme in \eqref{eq:schemeE gaussian} with $\la = \f{\al P_1}{\al P_1+1} \left(a + \sqrt{\f{\alb P_1}{P_2}}\right)$.
The bounds \eqref{eq:schemeE gaussian R1} and \eqref{eq:schemeE gaussian extra R1} are given by
\ea{
R_1 & \leq \Ccal  \lb \al \min\{1,|b|^2\}P_1\rb
}
in this case.

This scheme achieves capacity when \eqref{eq:schemeE gaussian R2} is larger than \eqref{eq:schemeE gaussian R1+R2}, i.e. the conditions in \eqref{eq:PDC} hold.
The conditions were determined in \cite{RTDjournal2} to prove the ``primary decodes cognitive'' regime for the Gaussian CIFC.
\end{proof}

\begin{rem}
The ``primary decodes cognitive regime'' for the CIFC is defined by conditions \eqref{eq:PDC} and condition \eqref{eq:strong interference} which is given by $|b| \geq 1$ in the Gaussian case.
Condition \eqref{eq:strong interference} is not required to prove capacity for the CIFC-CCM.
Thus, the capacity of the Gaussian CIFC-CCM for the regime with $|b| \leq 1$ is also given by Cor.~\ref{cor:Primary Decodes Cognitive Interference Regime for the CIFC-CCM}.
\end{rem}

\subsection{Capacity to Within a Constant Gap and a Constant Factor}
\label{sec:Capacity to Within a Constant Gap and a Constant Factor}

\begin{thm}{\bf Capacity to within 1 bits/s/Hz}
\label{eq:Capacity to within 1 bits/s/Hz}

For any Gaussian CIFC-CCM, the outer bound region in \eqref{eq:Outer Bound for the G-CIFC-CCM} can be achieved to within 1 bits/s/Hz.
\end{thm}

\begin{proof}
Consider  the assignment
\eas{
X_i     & \sim \Ncal (0, P_i)  \quad i \in \{1,2\}, \\
E[X_1 X_2 ] & =  \sqrt{\alb P_1 P_2}, \\
U_{1c}  & = X_1 + a X_2 + \Nt_1, \\
\Nt_1   & \sim \Ncal(0, \sgs_1),
}{}
for the scheme in \eqref{eq:scheme E}, than we have

\eas{
R_1         & \leq \log \lb  \sgs_{1} + \al P_1 \rb  - {\rm GAP_1}(\al),\label{eq:Const gap for the G-CIFC-CCM R1}\\
R_1         & \leq \Ccal \lb \alb |b|^2 P_1 \rb, \\
R_1+R_2     & \leq \Ccal \lb P_2 + |b^2 P_1| + 2 \sqrt{\alb |b|^2 P_1 P_2} \rb, \\
R_1+R_2     & \leq \Ccal \lb P_2 + |b^2 P_1| + 2 \sqrt{\alb |b|^2 P_1 P_2} \rb\notag\\
	    & \quad \quad- {\rm GAP_2}(\al),
}
where
\ea{
I(&Y_1; U_{1c})-I(U_{1c},X_2) = - H(U_{1c}|Y_1) + H(U_{1c}| X_2) \notag\\
& = - H(\Nt_1-N_1| X_1 + a X_2 + N_1) + H(X_1 + \Nt_1| X_2) \notag\\
& = \log \lb \sgs_{1} + \al P_1 \rb - \log \lb \sgs_{1} + \f {\VARbb[X_1 + a X_2]}{1+\VARbb[X_1+a X_2]}\rb  \notag\\
& = \log \lb \sgs_{1} + \al P_1 \rb- {\rm GAP_1}(\al)
}{}

and

\ea{
 - &{\rm GAP_2}(\al)
 = I(Y_1; U_{1c})-I(U_{1c},X_2)\notag\\
    &\quad \quad - \lb I(Y_2; U_{1c})-I(U_{1c},X_2)  \rb   \notag\\
& = - H(U_{1c}| Y_1) + H(U_{1c}| Y_2 ) \notag\\
& = - \log \lb \sgs_{1} + \f {\VARbb[X_1 + a X_2]}{1+\VARbb[X_1+a X_2]}\rb\notag\\
&\quad \quad  + H(X_1 + a X_2 + \Nt_1 | |b| X_1 + X_2 + N_2  )
}{}

Now fix $\sgs_{1}=1$ to obtain the outer bound expression.
Clearly $ {\rm GAP_1}(\al)> {\rm GAP_2}(\al)$ and ${\rm GAP_1}(\al)\leq 1$.
\end{proof}

The result in Th.~\ref{eq:Capacity to within 1 bits/s/Hz} also provides an alternative proof to the constant gap in \cite{Rini:ICC2010}
where a constant additive gap between inner and outer bound is proved using an achievable scheme with a private cognitive message in the ``strong interference'' regime.
As noted in \cite{RTDjournal2}, in the ``strong interference'' regime for the Gaussian channel the primary decoder can reconstruct the channel output at the cognitive decoder and thus decode the cognitive message.
For this reason it is counterintuitive to consider a scheme with a private cognitive message to achieve the outer bound in the ``strong interference'' regime.
Indeed the distance between inner and outer bound using the scheme in \eqref{eq:scheme E} is always smaller or equal to the distance using the scheme in \cite{Rini:ICC2010}.
Despite of this, the two schemes achieve the same gap from the $R_1$ bound in the outer bound expression and thus the overall bound between inner and outer bound is the same in the two cases.

\begin{thm}{\bf Capacity to within a Factor 2}

For any Gaussian CIFC-CCM, the outer bound region in \eqref{eq:Outer Bound for the G-CIFC-CCM} can be achieved to within a factor 2.
\end{thm}

\begin{proof}
Capacity of the Gaussian CIFC-CCM is known for the regime with $|b| \leq 1$.  The achievability of the outer bound to within  a factor $2$ for $|b|>1$  by a simple time division scheme is shown in \cite{RTDjournal2} for the Gaussian CIFC and can be directly applied to the Gaussian CIFC-CCM.
\end{proof}

\section{Numerical Simulations}
\label{sec:Numerical Simulations}

In the following we illustrate the results presented for the Gaussian CIFC-CCM in the previous section by means of numerical simulations.
We begin by comparing the outer bound in Cor.~\ref{cor:An Outer Bound for the Gaussian CIFC-CCM} with the outer bound in Cor.~\ref{cor:BC-DMS-based outer Bound for the Gaussian CIFC-CCM}:
in general it is not possible to simplify the expression in \eqref{eq:BC-DMS outer Bound for the G-CIFC-CCM} and determine when it is tighter than \eqref{eq:Outer Bound for the G-CIFC-CCM}.
To determine the convex hull of \eqref{eq:BC-DMS outer Bound for the G-CIFC-CCM} one needs to find the assignment of the parameters in  \eqref{eq:BC-DMS outer Bound for the G-CIFC-CCM parameters} that maximizes $R_1 + \mu R_2$ for all $\mu \in \Rbb^+$.
This problem does not have a simple closed form solution apart from some special cases such as the degraded channel and the  Z channel, in which $a=0$ \cite{Rini:ISIT2011:Z}.
Although a closed form solution cannot be determined, we can numerically simulate the outer bound in  Cor.~\ref{cor:BC-DMS-based outer Bound for the Gaussian CIFC-CCM} and conclude that this outer bound is indeed tighter than  Cor.~\ref{cor:An Outer Bound for the Gaussian CIFC-CCM} in a subspace of the parameter regime.
This is illustrated in Fig.~\ref{fig:Out DMS}: since the outer bound \eqref{eq:Outer Bound for the G-CIFC-CCM} does not depend on the channel parameter $a$, we fix the remaining channel parameters and plot the outer bound in \eqref{eq:BC-DMS outer Bound for the G-CIFC-CCM parameters} for increasing values of $a$.
The analytical evaluation of the outer bound in \eqref{eq:BC-DMS outer Bound for the G-CIFC-CCM parameters} is quite involved and thus the results in the plot are obtained by an exhaustive search over the parameter space.
From Fig.~\ref{fig:Out DMS} we conclude that the tightest outer bound is obtained by considering the intersection of the outer bounds in Cor.~\ref{cor:BC-DMS-based outer Bound for the Gaussian CIFC-CCM} and Cor.~\ref{cor:An Outer Bound for the Gaussian CIFC-CCM} as one outer bound does not strictly include the other for all rate pairs.

\begin{figure}
\centering
\includegraphics[width=1.0\columnwidth]{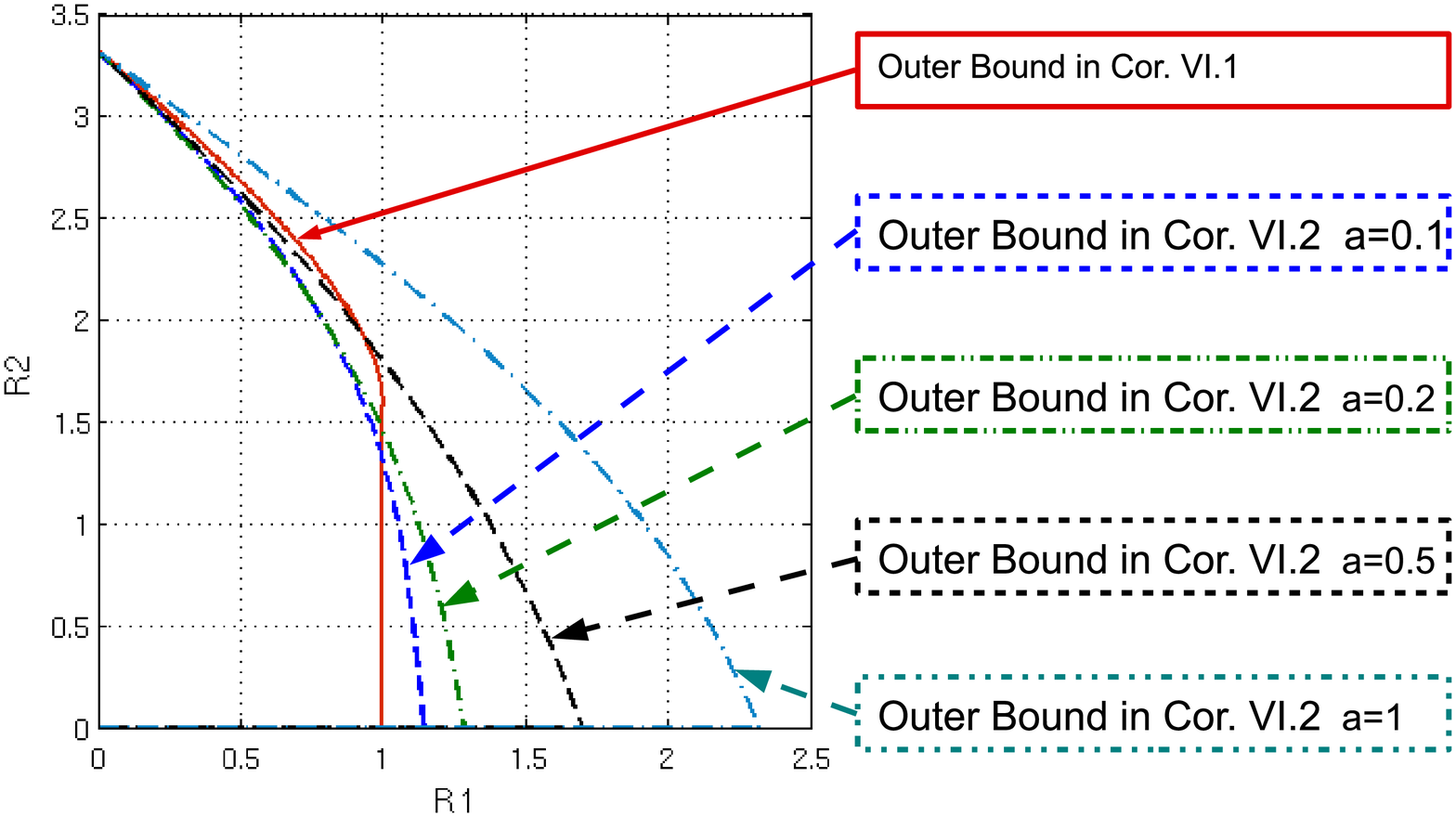}
\vspace*{-2 cm}
\caption{Outer bounds for Gaussian CIFC-CCM with $b=2,P_1=P_2=1$ and $a \in [0.1,0.2, 0.5, 1] $.}
\label{fig:Out DMS}
\end{figure}

\medskip

In Fig.~\ref{fig:Simu a01 b4} we compare outer  and inner bounds for a specific Gaussian CIFC-CCM.
In this regime the outer bound of Cor.~\ref{cor:An Outer Bound for the Gaussian CIFC-CCM} is tighter than the
outer bound in Cor.~\ref{cor:BC-DMS-based outer Bound for the Gaussian CIFC-CCM} for all the $R_1$ values in the interval $[0 \ldots \Ccal(1+P_1)]$.
The figure also compares the inner bound in Cor.~\ref{cor:Capacity for the Gaussian CIFC-CCM in the very strong interference regime} and the inner bound in Cor.~\ref{cor:Primary Decodes Cognitive Interference Regime for the CIFC-CCM} for the case of full interference cancellation and no interference cancellation.
In the latter achievable scheme full interference cancellation is obtained with the choice
\ea{
\lambda = \la_{\rm Costa}=\frac{\alpha P_1}{\alpha P_1 + 1} \left(a + \sqrt{\f{\alb P_1}{P_2}}\right),
}
as defined in \eqref{eq:la costa} while no interference cancellation is obtained with the choice $\la=0$.
The choice $\la=\la_{\rm Costa}$ maximizes the cognitive rate $R_1$ and has been shown to achieve capacity in the ``primary decodes cognitive'' regime of Th.~\ref{cor:Primary Decodes Cognitive Interference Regime for the CIFC-CCM}.
Interestingly the choice $\la=0$ outperforms the choice $\la=\la_{\rm Costa}$ in a large range of $R_1$ which indicates that, in this particular regime, having the primary receiver decode the cognitive codeword can be more easily performed when such codeword is not encoded against the interference.
The inner bound in Cor.~\ref{cor:Capacity for the Gaussian CIFC-CCM in the very strong interference regime} performs much worse than the other two schemes: in this achievable scheme both decoders are required to decode both codewords. Since the value of $a$ we consider is fairly small ($a=0.1$) having the cognitive receiver decode the primary codeword imposes strong restrictions on the rate $R_2$ and this results in a drastic loss of performance.
Note in particular that the point with the largest rate $R_2$ achieved by this scheme is given by
\ea{
(R_1,R_2) =\lb 0, \Ccal\left(|b| \sqrt{\alb P_1}+\sqrt{P_2}\right)^2 \rb .
}
This point is inferior to the point with the largest rate $R_2$ achieved by the other schemes and the outer bound, i.e.,
\ea{
(R_1,R_2) =\lb 0, \Ccal(|b|^2 P_1 + P_2 + 2 \sqrt{|b|^2 P_1 P_2}) \rb.
}

This is because the bound in \eqref{eq:scheme D Y1 bound} is active and reduces the largest achievable sum rate.
Despite of this, the scheme achieves capacity in the ``very strong interference'' regime of Cor.~\ref{cor:Capacity for the Gaussian CIFC-CCM in the very strong interference regime} where both the values of $a$ and $b$ are large with respect to the direct links.
\begin{figure}
\centering
\includegraphics[width=1.0\columnwidth]{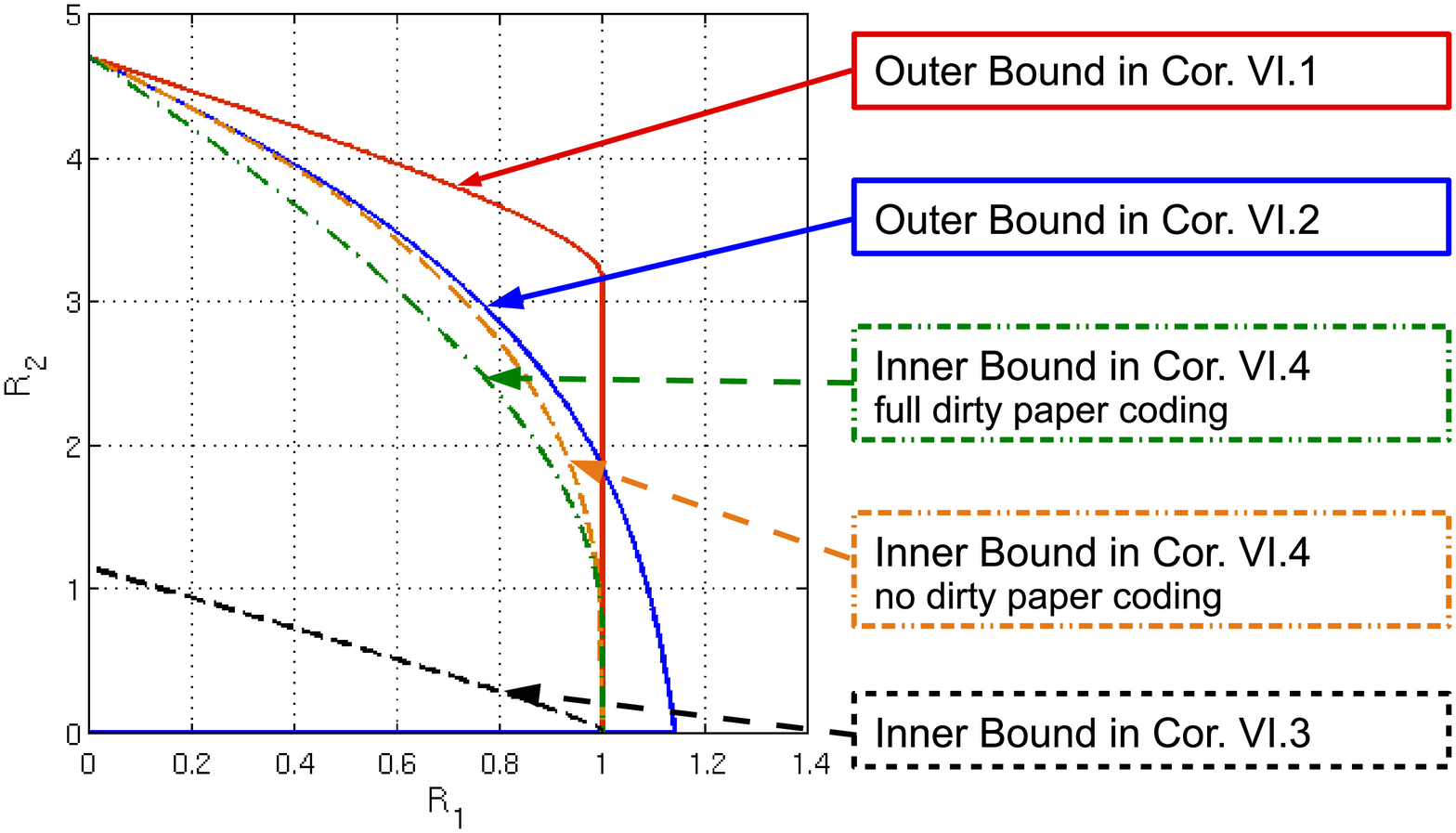}
\vspace*{-2 cm}
\caption{Inner and outer bounds for Gaussian CIFC-CCM with $a=0.1$, $b=4$ and $P_1=P_2=1$.}
\label{fig:Simu a01 b4}
\end{figure}

\medskip

Next, we compare the inner bound given in Cor.~\ref{cor:Primary Decodes Cognitive Interference Regime for the CIFC-CCM} with full interference cancellation that achieves capacity in the ``primary decodes cognitive'' regime with the inner bound from Cor.~\ref{cor:Capacity for the Gaussian CIFC-CCM in the very strong interference regime}. Fig.~\ref{fig:PCD VSI b2} depicts the achievable rates for different values of $a$ and the other channel parameters are fixed. The parameters are chosen such that capacity is not known, i.e., the conditions for the ``primary decodes cognitive'' regime in \eqref{eq:PDC} as well as for  ``very strong interference'' regime in \eqref{eq:very strong interference gaussian} are not fulfilled. 
In such regime the scheme from Cor.~\ref{cor:Primary Decodes Cognitive Interference Regime for the CIFC-CCM} outperforms the scheme from Cor.~\ref{cor:Capacity for the Gaussian CIFC-CCM in the very strong interference regime} for small values of $R_1$. 
This is due to the fact that \eqref{eq:scheme D Y1 bound} is always active and thus, the maximum achievable rate $R_2 = \Ccal(|b|^2 P_1 + P_2 + 2 \sqrt{|b|^2 P_1 P_2})$ is achieved by the scheme from Cor.~\ref{cor:Primary Decodes Cognitive Interference Regime for the CIFC-CCM} and not achieved by Cor.~\ref{cor:Capacity for the Gaussian CIFC-CCM in the very strong interference regime}. 
However, if both, $a$ and $\alpha$ are increasing, the scheme from Cor.~\ref{cor:Capacity for the Gaussian CIFC-CCM in the very strong interference regime} exceeds the achievable rate region of Cor.~\ref{cor:Primary Decodes Cognitive Interference Regime for the CIFC-CCM}.

\begin{figure}
\centering
\includegraphics[width=1.0\columnwidth]{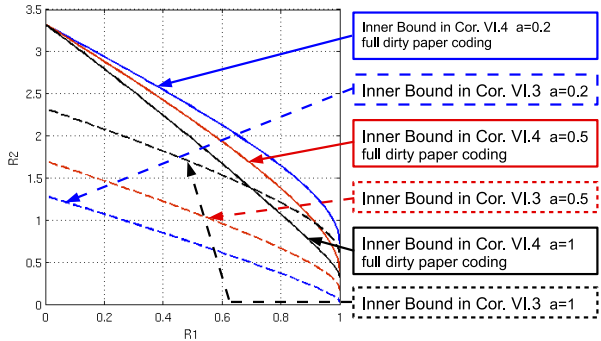}
\vspace*{-2 cm}
\caption{Comparison of the achievable schemes in Cor.~\ref{cor:Capacity for the Gaussian CIFC-CCM in the very strong interference regime} and Cor.~\ref{cor:Primary Decodes Cognitive Interference Regime for the CIFC-CCM} for with $b=2$ and $P_1=P_2=1$.}
\label{fig:PCD VSI b2}
\end{figure}

\medskip

In Fig.~\ref{fig:Simu constgap}, the scheme from Th.~\ref{eq:Capacity to within 1 bits/s/Hz} with $\sgs_{1}=1$ is illustrated for different cognitive transmit powers $P_1$. This scheme is used to prove the constant additive gap in Th. \ref{eq:Capacity to within 1 bits/s/Hz} and thus approaches capacity for large SNR
which, in the standard model of \eqref{eq:in/out gaussian CIFC-CCM} means  large transmit powers.
The figure shows that the scheme from Th.~\ref{eq:Capacity to within 1 bits/s/Hz} approaches capacity for increasing $P_1$ and small cognitive rates $R_1$.
However, for large rates $R_1$, the gap between inner and outer bounds on the $R_1$ coordinate, as given in \eqref{eq:Const gap for the G-CIFC-CCM R1}, approaches 1 bit/s/Hz:
\ea{
\underset{P_1 \rightarrow \infty}{\lim}{\rm GAP_1}(1) &= \underset{P_1 \rightarrow \infty}{\lim} \log \lb 1 + \f {P_1 + |a|^2 P_2}{1+ P_1 + |a|^2 P_2}\rb\notag\\
& = 1.
}

\begin{figure}
\centering
\includegraphics[width=1.0\columnwidth]{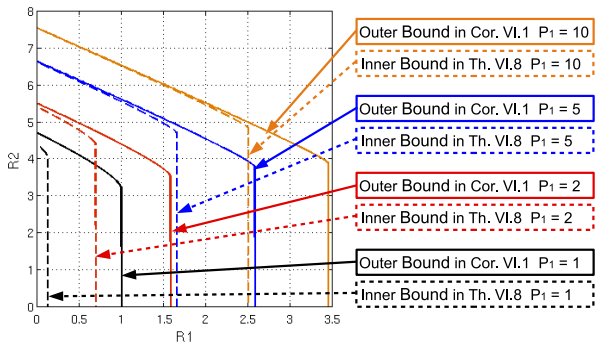}
\vspace*{-2 cm}
\caption{Constant gap approximation for Gaussian CIFC-CCM with $a=2$, $b=4$ and $P_2=1$.}
\label{fig:Simu constgap}
\end{figure}

\medskip

The inner bounds from Cor.~\ref{cor:Primary Decodes Cognitive Interference Regime for the CIFC-CCM} with full interference cancellation, from Cor.~\ref{cor:Capacity for the Gaussian CIFC-CCM in the very strong interference regime}, as well as from Th.~\ref{eq:Capacity to within 1 bits/s/Hz} with $\sgs_{1}=1$ are compared in Fig.~\ref{fig:Simu a1 b1.5} in a regime in which capacity is still unknown. 
In the considered setup none of the approaches outperforms the others over the whole region. For $R_1 = 0$ the scheme from Cor.~\ref{cor:Primary Decodes Cognitive Interference Regime for the CIFC-CCM} touches the outer bound and thus, outperforms both other schemes. The scheme from Th.~\ref{eq:Capacity to within 1 bits/s/Hz} can outperform the others only for medium rates $R_1$ whereas the scheme from Cor.~\ref{cor:Capacity for the Gaussian CIFC-CCM in the very strong interference regime} is superior for large rates $R_1$.

\begin{figure}
\centering
\includegraphics[width=1.0\columnwidth]{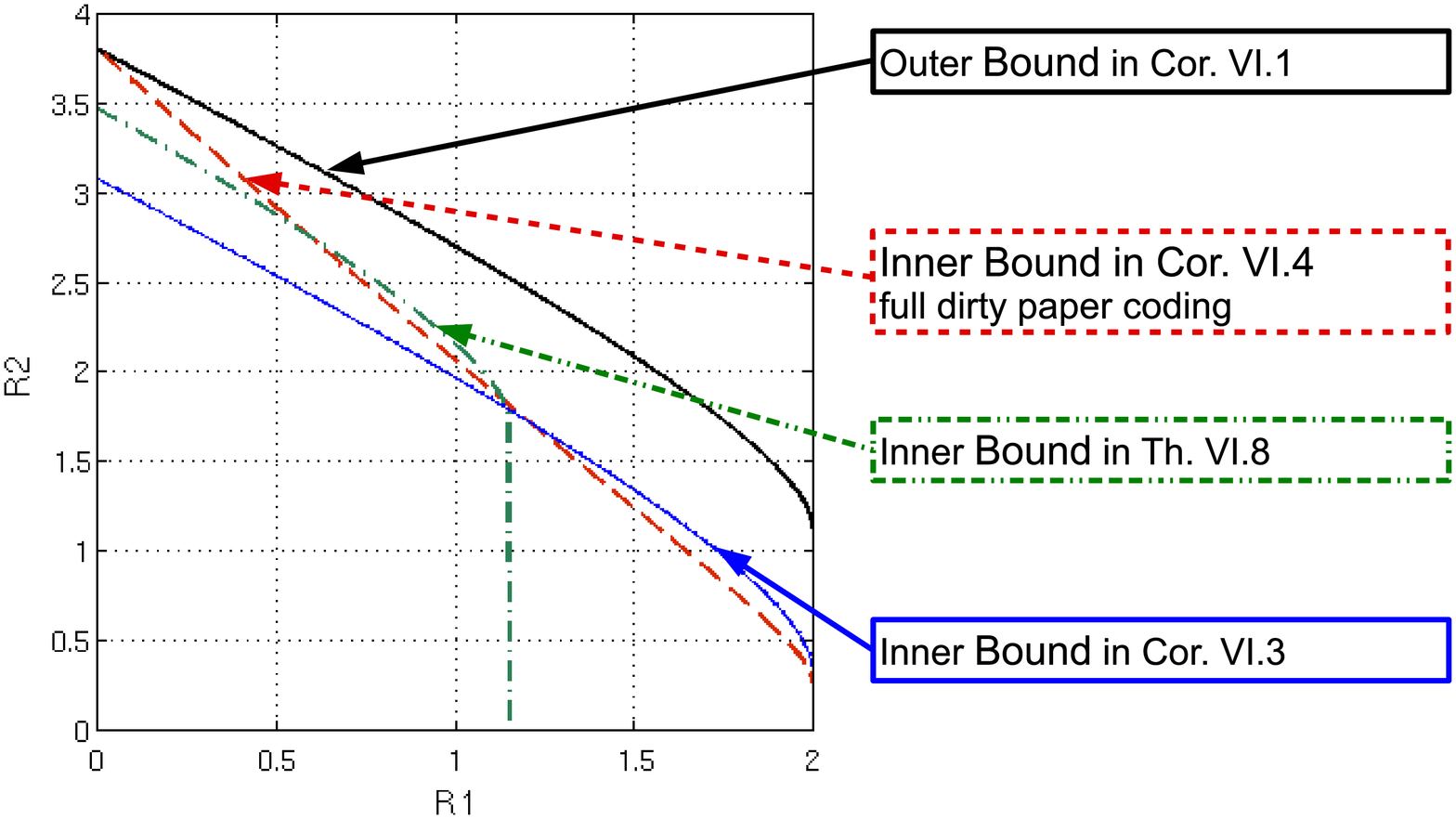}
\vspace*{-2 cm}
\caption{Comparison of different schemes for Gaussian CIFC-CCM with $a=1$, $b=1.5$, $P_1=3$ and $P_2=1$.}
\label{fig:Simu a1 b1.5}
\end{figure}

\medskip

A plot of the capacity results available for the Gaussian CIFC-CCM is depicted in Fig.~\ref{fig:Gaussian Plot} for $P_1=10$ and $P_2 = 10$: in the $a \times b$ plane we plot the ``very strong interference'' regime (area with circles) and the ``primary decodes cognitive regime'' (area with dots).

\begin{figure}
\centering
\includegraphics[width=1.0\columnwidth]{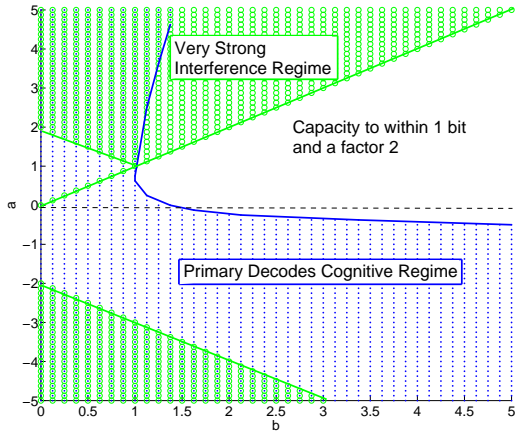}
\caption{Capacity for Gaussian CIFC-CCM.}
\label{fig:Gaussian Plot}
\end{figure}

\section{Conclusion}
\label{sec:conclusion}

The paper studies a variation of the classical cognitive interference channel in which the primary receiver decodes both messages.
This channel is related to the cognitive interference channel in the ``strong interference'' regime and many results for the cognitive interference channel apply to the model under consideration.
For this channel, we give outer bounds to the capacity region based with a different number of auxiliary random variable.
A general inner bound, comprising rate splitting, superposition coding and binning is introduced. 
To prove capacity in special regimes, two simple sub-schemes, one using only superposition coding while the other applies solely binning, are given. 
Using these schemes, we show capacity in the ``very strong interference'' regime, where there is no rate loss  in having both receivers decode both messages, as well as capacity in the ``primary decodes cognitive'' regime for discrete memoryless channels.
We also derive the capacity for the semi-deterministic case, where the cognitive output is a deterministic function of the channel inputs.
Furthermore, we determine capacity for the Gaussian case to within a constant additive gap of one bit/s/Hz and to within a constant multiplicative factor of two. %

\section*{Acknowledgments}
The work of Stefano Rini was supported by the German Ministry of Education and Research in the framework of an Alexander von Humboldt Professorship.
The work of Carolin Huppert was supported by the German research council ``Deutsche Forschungsgemeinschaft'' (DFG) under Grant~Hu~1835/3.

\bibliographystyle{unsrt}
\bibliography{steBib1}

\appendix
\section{Proof of Corollary~\ref{cor:rate-sharing}}
\label{app:Proof of Cor rate-sharing}
To make the notation more compact we employ in the following the short-hand notation
\eas{
&I_{Y_1}  = I(Y_1 ; U_{2c}, U_{1c})-I(U_{1c}; X_2| U_{2c}),\\
&I_{Y_1| U_{2c}}  = I(Y_1 ; U_{1c} | U_{2c}) - I(U_{1c}; X_2| U_{2c}),\\
&I_{Y_2}  =  I(Y_2 ; U_{2c}, U_{1c}, X_2 , U_{2pb} ),\\
&I_{Y_2| U_{2c}}  =  I(Y_2 ; U_{1c}, X_2 , U_{2pb} | U_{2c}),\\
&I_{Y_2| U_{2c},U_{1c}}  =  I(Y_2 ;  X_2 , U_{2pb} | U_{2c}, U_{1c})\notag\\
	  & \quad \quad + I(U_{1c}; X_2 | U_{2c}),\\
&I_{Y_2| U_{2c},X_2}  = I(Y_2 ; U_{1c} , U_{2pb} | U_{2c}, X_2),\\
&I_{Y_2| U_{2c},U_{1c},X_2}  = I(Y_2 ;  U_{2pb} | U_{2c}, U_{1c}, X_2).
}

Consider the proof of Th.~\ref{th:Inner Bounds for the CIFC-CCM} and note that if the rate vector
$\lb \tilde{R}_{2c}, \tilde{R}_{1c}, \tilde{R}_{2p}, \tilde{R}_{2pb}\rb$
is achievable, then also the rate vector
$\lb R_{2c}, R_{1c}, R_{2p}, R_{2pb}\rb$,
with
\eas{
R_{2c} = \tilde{R}_{2c}-\Delta_{2p}^{(2)}-\Delta_{2pb}^{(2)}+\Delta_{2c}^{(1)},\\
R_{1c} = \tilde{R}_{1c}-\Delta_{2p}^{(1)}-\Delta_{2pb}^{(1)}-\Delta_{2c}^{(1)},\\
R_{2p} = \tilde{R}_{2p}+\Delta_{2p}^{(1)}+\Delta_{2p}^{(2)},\\
R_{2pb} = \tilde{R}_{2pb}+\Delta_{2pb}^{(1)}+\Delta_{2pb}^{(2)},
}
is achievable as long as the latter rates $R_{2c}$, $R_{1c}$, $R_{2p}$, $R_{2pb}$ are still positive, that is as long as
\eas{
\tilde{R}_{2c}  & \geq \Delta_{2p}^{(2)} + \Delta_{2pb}^{(2)}-\Delta_{2c}^{(1)}, \\
\tilde{R}_{1c} & \geq \Delta_{2p}^{(1)} +\Delta_{2pb}^{(1)}+\Delta_{2c}^{(1)},
}{\label{eq:positive condition}}
as shown in \cite{hajek1979evaluation} and \cite{liang2007rate}.

The above implies that the achievability of the region in \eqref{eq:achievable region before FME} also implies the achievability of the region
\eas{
&R_{2c},R_{1c},R_{2p},R_{2pb}  \geq 0, \\
&\Delta_{2p}^{(2)},\Delta_{2pb}^{(2)}, \Delta_{2p}^{(1)},\Delta_{2pb}^{(1)},\Delta_{2c}^{(1)}  \geq 0, \\
&R_{2c}+\Delta_{2p}^{(2)} + \Delta_{2pb}^{(2)}   \geq \Delta_{2c}^{(1)},
\label{eq:achievable region before FME RS positive 1} \\
&R_{2p}   \geq \Delta_{2p}^{(1)} + \Delta_{2p}^{(2)}
\label{eq:achievable region before FME RS positive 2} \\
&R_{2pb}   \geq \Delta_{2pb}^{(1)} + \Delta_{2pb}^{(2)}
\label{eq:achievable region before FME RS positive 3} \\
&R_{2c} + R_{1c} + \Delta_{2p}^{(2)}+ \Delta_{2pb}^{(2)}+\Delta_{2p}^{(1)}+\Delta_{2pb}^{(1)}   \leq I_{Y_1},
\label{eq:achievable region before FME RS a}\\
&R_{1c} +\Delta_{2p}^{(1)}+\Delta_{2pb}^{(1)} + \Delta_{2c}^{(1)}  \leq I_{Y_1|U_{2c}},
\label{eq:achievable region before FME RS b}\\
&R_{2c}+R_{1c}+R_{2p}+R_{2pb}  \leq I_{Y_2| U_{2c},U_{1c}},
\label{eq:achievable region before FME RS c}\\
&R_{1c}+R_{2p}+R_{2pb} + \Delta_{1c}^{(1)} - \Delta_{2p}^{(2)}- \Delta_{2pb}^{(2)}  \leq I_{Y_2 | U_{2c}},
\label{eq:achievable region before FME RS d}\\
&R_{2p}+R_{2pb}- \Delta_{2p}^{(2)}- \Delta_{2pb}^{(2)} -  \Delta_{2p}^{(1)}- \Delta_{2pb}^{(1)}   \notag\\
&\quad \leq I_{Y_2 | U_{2c}, U_{1c}} + I(U_{1c}; X_2 | U_{2c}),
\label{eq:achievable region before FME RS e}\\
&R_{1c}+R_{2pb} + \Delta_{2c}^{(1)} + \Delta_{2p}^{(1)} - \Delta_{2pb}^{(2)}  \leq I_{Y_2 | U_{2c}, X_2},
\label{eq:achievable region before FME RS f}\\
&R_{2pb}- \Delta_{2pb}^{(1)} - \Delta_{2pb}^{(2)}  \leq I_{Y_2 | U_{2c}, U_{1c}, X_2},
\label{eq:achievable region before FME RS g}
}{\label{eq:achievable region before FME RS}}
where \eqref{eq:achievable region before FME RS positive 1},\eqref{eq:achievable region before FME RS positive 2} and \eqref{eq:achievable region before FME RS positive 3} result from the condition $\tilde{R}_{2c}, \tilde{R}_{1c}, \tilde{R}_{2p}, \tilde{R}_{2pb} \geq 0 $.
We now proceed with the FME of all the $\Delta$s to obtain a compact representation of the achievable region only in terms of $R_{2c},R_{1c},R_{2p}$ and $R_{2b}$.
First of all note that $\Delta_{2c}^{(1)}$ always decreases the upper rate bounds in \eqref{eq:achievable region before FME RS}. This implies that the
largest achievable region is obtained by simply setting $\Delta_{2c}^{(1)}=0$.
The FME of all the remaining $\Delta$s is algorithmically complex: for this reason we proceed in eliminating $\Delta$s is successive steps.
We begin by eliminating $\Delta_{2p}^{(1)}$ which yields the region
\eas{
&R_{2c},R_{1c},R_{2p},R_{2pb}  \geq 0, \\
&\Delta_{2p}^{(2)},\Delta_{2pb}^{(2)} ,\Delta_{2pb}^{(1)}  \geq 0, \\
&R_{2p}   \geq  \Delta_{2p}^{(2)} \\
&R_{2pb}   \geq \Delta_{2pb}^{(1)} + \Delta_{2pb}^{(2)} \\
&R_{2pb} -\Delta_{2pb}^{(1)} + \Delta_{2pb}^{(2)}   \leq I_{Y_2| U_{2c},U_{1c},X_2},
\label{rate sharing step1 a}\\
&R_{1c}  + \Delta_{2pb}^{(1)} - \Delta_{2pb}^{(2)} \leq I_{Y_1| U_{2c}},
\\
&R_{1c}+R_{2pb}- \Delta_{2pb}^{(2)}   \leq I_{Y_2| U_{2c},X_2}\label{rate sharing step1 c},
\\
&R_{1c}+R_{2p}+R_{2pb}  - \Delta_{2p}^{(2)}- \Delta_{2pb}^{(2)}    \notag\\
& \quad \leq \min \lcb I_{Y_2| U_{2c}}, I_{Y_1| U_{2c}} + I_{Y_2| U_{2c},U_{1c}} \rcb ,
\\
&R_{1c}+R_{2p}+2 R_{2pb} - \Delta_{2p}^{(2)}- 2 \Delta_{2pb}^{(2)}- \Delta_{2pb}^{(1)}   \notag\\
& \quad \leq I_{Y_2| U_{2c},X_2} + I_{Y_2| U_{2c},U_{1c}},
\\
&R_{1c}+R_{2c} + \Delta_{2p}^{(2)}+ \Delta_{2pb}^{(2)} +\Delta_{2pb}^{(1)}     \leq I_{Y_1} \label{rate sharing step1 g},
\\
&R_{1c}+R_{2c}+R_{2p}+R_{2pb}   \notag\\
& \quad  \leq \min \lcb I_{Y_2},I_{Y_1}  + I_{Y_2| U_{2c},U_{1c}} \rcb ,
\label{rate sharing step1 i}
}{\label{eq:achievable region FME step1} }

In the next steps we first eliminate $\Delta_{2p}^{(2)}$ and successively $\Delta_{2pb}^{(1)}$ and $\Delta_{2pb}^{(2)}$ again by means of FME which gives

\eas{
&R_{2c},R_{1c},R_{2p},R_{2pb}  \geq 0, \\
&R_{1c}   \leq \min \lcb I_{Y_1 | U_{2c}},I_{Y_2 | U_{2c}, X_2}  \rcb \\
&R_{1c}+R_{2c}   \leq I_{Y_1}, \\
&R_{1c}+R_{2c}+R_{2pb}   \leq I_{Y_1} + I_{Y_2 | U_{2c},U_{1c}, X_2}, \\
&R_{1c}+R_{2c}+R_{2p}+R_{2pb}   \leq \min \lcb I_{Y_2}, I_{Y_1} + I_{Y_2| U_{2c},U_{1c}} \rcb \\
&2 R_{1c}+R_{2c}+ R_{2pb}   \leq I_{Y_1} + I_{Y_2 | U_{2c}, X_2}, \\
&2 R_{1c}+R_{2c}+R_{2p}+R_{2pb}   \leq I_{Y_1} + I_{Y_2 | U_{2c}}.
}{\label{eq:achievable region FME step4} }

From the FME of $R_2$ in \eqref{eq:achievable region FME step4} and
\eqref{eq:rate splitting definition}  we finally obtain the achievable rate region \eqref{eq:inner bound}.
This concludes the proof.

\end{document}